\documentclass[a4paper,11pt]{amsart}
\usepackage[pdfdisplaydoctitle=true,colorlinks=true,urlcolor=blue,citecolor=blue,linkcolor=blue,pdfstartview=FitH,pdfpagemode=None,bookmarksnumbered=true]{hyperref}
\usepackage{cmap} 

\newtheorem{theorem}{Theorem}[section]

\newtheorem{lemma}[theorem]{Lemma}
\newtheorem{proposition}[theorem]{Proposition}

\theoremstyle{definition}
\newtheorem{definition}[theorem]{Definition}

\theoremstyle{remark}
\newtheorem{remark}[theorem]{Remark}
\newtheorem{example}[theorem]{Example}

\numberwithin{equation}{section}

\newcommand{\Real}{\mathbb R}

\newcommand{\Circle}{\mathbb{S}^{1}}
\newcommand{\KK}{\mathcal{K}}

\newcommand{\g}{\mathfrak{g}}
\newcommand{\gstar}{\mathfrak{g}^{*}}

\newcommand{\VectM}{\mathrm{Vect}(M)}
\newcommand{\CM}{\mathrm{C}^{\infty}(M)}

\newcommand{\DiffS}{\mathrm{Diff}(\mathbb{S}^{1})}
\newcommand{\VectS}{\mathrm{Vect}(\mathbb{S}^{1})}
\newcommand{\VectSstar}{\mathrm{Vect}^{*}(\mathbb{S}^{1})}
\newcommand{\CS}{\mathrm{C}^{\infty}(\mathbb{S}^{1})}

\newcommand{\poisson}[2]{\left\{\,#1,#2\right\}}
\newcommand{\brackets}[2]{\left[#1,#2\right]}

\hypersetup{
    pdftitle={Integrability of invariant metrics on the diffeomorphism group of the circle}, 
    pdfauthor={Adrian Constantin, Boris Kolev}, 
    pdfsubject={2000 Mathematics Subject Classification: 35Q35, 35Q53, 37K10, 37K65}, 
    pdfkeywords={Bi-Hamiltonian formalism, Diffeomorphisms group of the circle, Camassa-Holm equation} 
    }

\begin{document}

\title[Integrability of invariant metrics]{Integrability of invariant metrics on the diffeomorphism group of the circle} %

\author[A. Constantin]{Adrian Constantin}%
\address{Trinity College, Department of Mathematics, Dublin 2, Ireland}%
\email{adrian@maths.tcd.ie}%

\author[B. Kolev]{Boris Kolev}%
\address{CMI, 39 rue F. Joliot-Curie, 13453 Marseille cedex 13, France }%
\email{boris.kolev@up.univ-mrs.fr}%

\thanks{This research was supported by the SFI-grant 04/BR/M0042}%

\date{July 2005}%


\begin{abstract}
Each $H^k$ Sobolev inner product ($k \ge 0$) defines a Hamiltonian vector
field $X_k$ on the regular dual of the Lie algebra of the diffeomorphism
group of the circle.
We show that only $X_0$ and $X_1$ are bi-Hamiltonian
relatively to a modified Lie-Poisson structure.
\end{abstract}

\maketitle


\section{Introduction}

Often motions of inertial mechanical systems are described in Lagrangian variables by paths on a configuration space $G$ that is a Lie group. The velocity phase space is the tangent bundle $TG$ and the kinetic energy
\begin{equation*}
    \KK = \frac{1}{2} <v,v>
\end{equation*}
for $v \in TG$. For example, in continuum mechanics the state of a system at time $t \ge 0$ can be specified by a diffeomorphism $x \mapsto \varphi(t,x)$ of the ambient space, giving the configuration of
the particles with respect to their initial positions at time
$t=0$. Here $x$ is a label identifying a particle, taken to be the
position of the particle at time $t=0$ so that $\varphi(0,x)=x$. In
this setting $G$ would be the group of diffeomorphisms. The
material (Lagrangian) velocity field is given by $(t,x) \mapsto \varphi_t(t,x)$
while the spatial (Eulerian) velocity field is
$u(t,y)=\varphi_t(t,x)$, where $y=\varphi(t,x)$, i.e. $u=\varphi_t \circ
\varphi^{-1}$. Observe that for any fixed time-independent
diffeomorphism $\eta$, the spatial velocity field $u=\varphi_t \circ
\varphi^{-1}$ along the path $t \mapsto \varphi(t)$ remains unchanged
if we replace this path by $t \mapsto \varphi(t) \circ \eta$. This
right-invariance property suggests to extend the kinetic energy
$\KK$ by right translation to a right-invariant Lagrangian
$\KK: TG \to \Real$, obtaining a Lagrangian system on
$G$. The length of a path $\{\varphi(t)\}_{t \in [a,b]}$ in $G$
is defined as
\begin{equation*}
    \mathfrak{l}(\varphi) = \int_{a}^{b} < \varphi_t,\varphi_t >^{1/2}\,dt .
\end{equation*}
The Least Action Principle
holds if the equation of motion is the geodesic equation.
The set $\DiffS$ of all smooth orientation-preserving diffeomorphisms of the circle represents
the configuration space for the spatially periodic motion of inertial
one-dimensional mechanical systems. $\DiffS$
is an infinite dimensional Lie group, the group operation being composition
\cite{Ham82} and its Lie algebra $\VectS$ being the space of all smooth vector fields on $\Circle$ cf. \cite{Sch87}. On the
regular (or $L^{2}$) dual $\VectSstar$ of the Lie algebra $\VectS$ there are some affine
canonical Lie-Poisson structures, called
\emph{modified Lie-Poisson structures}, which are all compatible. On the other hand,
we can consider on the regular dual $\VectSstar$ a
countable family $\{X_k\}_{k \ge 0}$ of Hamiltonian vector fields defined by Sobolev inner products. The importance of these
inner products lies in that each gives rise via right translation to a
geodesic flow on $\DiffS$, the Riemannian exponential
map of which defines a local chart
for every $k \ge 1$ cf. \cite{CK03} - a property which
fails for the Lie group
exponential map \cite{Ham82,Mil84} as well as for the Riemannian exponential map
if $k=0$ \cite{CK02}. In this paper we show that the Hamiltonian vector field
$X_k$ is bi-Hamiltonian relatively to a modified Lie-Poisson
structure if and only if $k \in \{0,1\}$.


\section{Preliminaries}

In this section, we review some fundamental aspects of finite
dimensional smooth Poisson manifolds.

\begin{definition}
A \emph{symplectic manifold} is a pair $(M,\omega)$, where $M$ is a manifold
and $\omega$ is a closed nondegenerate $2$-form on $M$, that is
$d\omega=0$ and for each $m \in M$, $\omega_m:T_mM \times T_mM \to
\Real$ is a continuous bilinear skew-symmetric map such that the induced
linear map $\tilde{\omega}_v: T_mM \to T_m^\ast M$ defined by
$\tilde{\omega}_v(w)=\omega(v,w)$ is an isomorphism for all $v \in T_mM$.
\end{definition}

\begin{example}\label{exa:SymplecticCase}
In the general study of variational problems, extensive use is made
of the canonical symplectic structure on the cotangent bundle
$T^{\ast} M$ (representing the phase space) of the manifold $M$
(representing the configuration space). This symplectic form is
given in any local trivialization
$(q,p) \in U \times \Real^{n} \subset \Real^{n}\times
\Real^{n}$ of $T^{\ast} M$ by
\begin{equation*}
    \omega_{(q,p)}\Big( (Q,P), \, (\tilde{Q},
\tilde{P}) \Big) =
    \tilde{P}\cdot Q - P\cdot \tilde{Q},\qquad (Q,P),\,(\tilde{Q},\tilde{P})
\in \Real^n \times \Real^n.
\end{equation*}
\end{example}

Since a symplectic form $\omega$ is nondegenerate, it induces an
isomorphism
\begin{equation}\label{equ:flatisomorphism}
    \flat :TM  \to T^{\ast}M,\quad X  \mapsto  X^{\flat} ,
\end{equation}
defined via $X^{\flat}(Y) = \omega(X,Y)$. The \emph{symplectic
gradient} $X_{f}$ of a function $f$ is defined by the relation
$X^{\flat}_{f} = -df$. The inverse of the isomorphism $\flat$
defines a skew-symmetric bilinear form $W$ on the cotangent space of
$M$. This bilinear form $W$ induces itself a bilinear mapping on
$\CM$, the space of smooth functions $f: M \to \Real$, given by
\begin{equation}\label{equ:PoissonBivector}
    \poisson{f}{g} = W(df,dg) = \omega(X_{f},X_{g}),
    \quad f,g \in \CM ,
\end{equation}
and called the \emph{Poisson bracket} of the functions $f$ and $g$.

\begin{example}
In Example~\ref{exa:SymplecticCase}, the Poisson bracket is given by
\begin{equation}\label{equ:even}
    \poisson{f}{g} = \sum_{i=1}^{n} \Big( \frac{\partial f}{\partial
q_{i}} \,
    \frac{\partial g}{\partial p_{i}} - \frac{\partial f}{\partial
p_{i}} \,
    \frac{\partial g}{\partial q_{i}}\Big).
\end{equation}
\end{example}

The observation that a bracket like \eqref{equ:even} could be introduced on
$\CM$ for a smooth manifold $M$, without the use of a symplectic form,
leads to the general notion of a \emph{Poisson structure} \cite{Lic77}.

\begin{definition}
A \emph{Poisson structure} on a $C^{\infty}$ manifold $M$ is a
skew-symmetric bilinear mapping $(f,g)\mapsto \poisson{f}{g}$ on the
space $\CM$, which
satisfies the \emph{Jacobi identity}
\begin{equation}\label{equ:JacobiIdentity}
    \poisson{\poisson{f}{g}}{h} + \poisson{\poisson{g}{h}}{f} +
    \poisson{\poisson{h}{f}}{g} = 0,
\end{equation}
as well as the \emph{Leibnitz identity}
\begin{equation}\label{equ:LeibnitzIdentity}
    \poisson{f}{gh} = \poisson{f}{g}h + g\poisson{f}{h}.
\end{equation}
\end{definition}

When the Poisson structure is induced by a symplectic structure
$\omega$, the \emph{Leibnitz identity} is a direct consequence of
\eqref{equ:PoissonBivector}, whereas the \emph{Jacobi
identity}~\eqref{equ:JacobiIdentity} corresponds to the condition
$d\omega = 0$ satisfied by the symplectic form $\omega$. In the
general case, the fact that the mapping $g\mapsto\poisson{f}{g}$
satisfies \eqref{equ:LeibnitzIdentity} means that it is a derivation
of $\CM$. Each derivation on $\CM$ for a
$C^{\infty}$ manifold (even in the infinite dimensional case cf. \cite{AMR88})
corresponds to a smooth vector field, that is, to each $f \in \CM$
is associated a vector field $X_{f}:M \to TM$, called the
\emph{Hamiltonian vector field} of $f$, such that
\begin{equation}\label{equ:DefHamiltonian}
    \poisson{f}{g} = X_{f} \cdot g = dg . X_{f},
\end{equation}
where $dg.X_f=L_{X_f}g$ is the \emph{Lie derivative} of $g$
along $X_f$. Conversely, a vector field $X: M \to TM$ on a
Poisson manifold $M$ is said to be
\emph{Hamiltonian} if there exists a function $f$ such that
$X=X_{f}$.

Recall \cite{Sch87} that for a smooth vector field $X:M \to TM$,
the Lie derivative operator $L_X:\CM \to \CM$ acts on
smooth functions $g:M \to \Real$ with
differentials $dg:M \to T^\ast M$ by
$(L_Xg)(m)=dg(m)\cdot X(m)$ for $m \in M$. The space $\VectM$
of smooth vector fields on $M$ and the space of operators
$\{L_X:\ X \in \VectM\}$ are isomorphic as real vector spaces,
the linear isomorphism between them being $X \mapsto L_X$ \cite{AMR88}. Therefore
the elements of $\VectM$ can be regarded as operators on
$\CM$ via $X\cdot f=L_Xf$, forming a Lie
algebra if endowed with the bracket $[X,Y]=L_X \circ L_Y-L_Y \circ L_X$.
Notice that \eqref{equ:JacobiIdentity} yields
\begin{equation}\label{equ:bracket}
    \brackets{X_{f}}{X_{g}}=X_{\poisson{f}{g}}.
\end{equation}
From \eqref{equ:bracket} it follows (see \cite{Sch87}) that $g \in
\CM$ is a constant of motion for $X_f$ if and only if
$\poisson{f}{g}=0$.

Jost \cite{Jos64} pointed out that, just like a derivation on $\CM$
corresponds to a vector field, a bilinear bracket
$\poisson{f}{g}$ satisfying the Leibnitz
rule~\eqref{equ:LeibnitzIdentity} corresponds to a skew-symmetric
bilinear form on $TM$. That is, there exists a $C^{\infty}$
tensor field $W\in\Gamma(\bigwedge^{2}TM)$, called the
\emph{Poisson bivector} of $(M,\poisson{\cdot}{\cdot})$, such that
\begin{equation*}
    \poisson{f}{g}=W(df,dg).
\end{equation*}
Using the unique local extension of the
Lie bracket of vector fields to skew-symmetric multivector fields,
called the \emph{Schouten-Nijenhuis bracket} \cite{Vai94}, the condition
\eqref{equ:JacobiIdentity} becomes
\begin{equation}\label{equ:SchoutenNijenhuis}
    \brackets{W}{W}=0.
\end{equation}
Conversely, any $W\in\Gamma(\bigwedge^{2}TM)$ that satisfies
\eqref{equ:SchoutenNijenhuis} induces a Poisson structure on $M$
via \eqref{equ:PoissonBivector}. The only condition that must be satisfied
by $W$ is \eqref{equ:SchoutenNijenhuis} since \eqref{equ:LeibnitzIdentity}
holds automatically. A Poisson structure on $M$ is therefore
equivalent to a bivector $W$
that satisfies \eqref{equ:SchoutenNijenhuis}. This induces a homomorphism
\begin{equation}\label{equ:SharpHomomorphism}
    \#:T^{\ast}M  \to TM,\quad \alpha  \mapsto  \alpha^{\#},
\end{equation}
such that $\beta(\alpha^{\#})=W(\beta ,\alpha)$ for every $\beta\in
T^{\ast}M$. Notice that for $f\in \CM$ we have
$(df)^{\#}=X_{f}$. If the homomorphism \eqref{equ:SharpHomomorphism}
is an isomorphism we call the Poisson structure
\emph{nondegenerate}. A nondegenerate Poisson structure on $M$ is
equivalent to a symplectic structure where the symplectic form
$\omega$ is just $\# W$, the closedness condition corresponding to
the Jacobi identity \cite{Vai94}.

\begin{remark}
The notion of a Poisson manifold is more general than that of a symplectic
manifold. For example, in the symplectic
case the Poisson bracket satisfies
the additional property that $\poisson{f}{g}=0$ for all $g \in \CM$
only if $f \in \CM$ is constant, whereas for Poisson manifolds
such non-constant functions $f$ might exist, in which case they are called
\emph{Casimir functions}. To highlight this, notice that by
Darboux' theorem \cite{Sch87} a finite dimensional symplectic manifold $M$
has to be even dimensional and locally there are coordinates
$\{q_1,...,q_n, p_1,...,p_n\}$ such that $\poisson{f}{g}$ is given
by \eqref{equ:even}. On the other hand, on $M=\Real^{2n+1}$ with coordinates
$\{q_1,...,q_n, p_1,...,p_n,\zeta\}$ we determine a Poisson structure
defining the Poisson bracket of $f,\,g \in C^\infty(\Real^{2n+1})$ by
the same formula \eqref{equ:even}. Notice that any
$f \in C^\infty(\Real^{2n+1})$ which depends only on $\zeta$ is a
Casimir function.
\end{remark}

Two Poisson bivectors $W_{1}$ and $W_{2}$ define \emph{compatible}
Poisson structures if
\begin{equation}\label{equ:CompatibilityCondition}
    \brackets{W_{1}}{W_{2}}=0.
\end{equation}
This is equivalent to say that for any $\lambda,\mu\in\Real$,
\begin{equation*}
    \poisson{f}{g}_{\lambda,\,\mu} = \lambda\poisson{f}{g}_{1}+\mu\poisson{f}{g}_{2}
\end{equation*}
is also a Poisson bracket. On a manifold $M$ equipped with two
compatible Poisson structures, a vector field $X$ is said to be
(formally) \emph{integrable} or
\emph{bi-Hamiltonian} if it is Hamiltonian for both structures.

On a symplectic manifold $(M,\omega)$, a necessary condition for a
vector field $X$ to be Hamiltonian is that $L_{X}\omega=0$ \cite{Sch87}.
A similar criterion exists for a Poisson manifold $(M,W)$. It is
instructive for later considerations to present a short proof of
this known result.

\begin{proposition}
On a Poisson manifold $(M,W)$ a necessary condition for a vector
field $X$ to be \emph{Hamiltonian} is
\begin{equation}\label{equ:HamiltonianCriterium}
    L_{X}W = 0.
\end{equation}
\end{proposition}

\begin{proof}
If $X$ is Hamiltonian, there is a function $h\in \CM$ such
that $X = X_{h}$. Let $f$ and $g$ be arbitrary smooth functions on
$M$. We have
\begin{equation*}
    L_{X}W \, (df,dg) = L_{X} \left( W(df,dg) \right) - W \left( L_{X}
    df,dg \right) - W \left( df,L_{X}dg \right) .
\end{equation*}
But $L_{X_{h}}f = \poisson{h}{f}$ and  $L_{X_{h}}df = dL_{X_{h}}f =
d\poisson{h}{f}$. Therefore
\begin{align*}
    L_{X}W(df,dg) & = L_{X}\poisson{f}{g} - W\left( d\poisson{h}{f},dg
    \right) - W\left( df,d\poisson{h}{g} \right) \\
              & = \poisson{h}{\poisson{f}{g}} -
    \poisson{\poisson{h}{f}}{g} - \poisson{f}{\poisson{h}{g}} .
\end{align*}
This last expression equals zero because of the Jacobi identity.
\end{proof}

The fundamental example of a non-symplectic Poisson structure is the
\emph{Lie-Poisson structure} on the dual
$\gstar$ of a Lie algebra $\mathfrak{g}$.

\begin{definition}
On the dual space $\gstar$ of a Lie algebra $\g$
of a Lie group $G$, there is a Poisson structure defined by
\begin{equation}\label{equ:LiePoisson}
    \poisson{f}{g}(m) = m(\brackets{d_{m}f}{d_{m}g})
\end{equation}
for $m\in\gstar$ and $f,g \in
C^{\infty}(\gstar)$, called the \emph{canonical
Lie-Poisson structure}~\footnote{Here, $d_{m}f$, the differential of a function $f \in C^\infty(\gstar)$ at $m \in
\gstar$ is to be understood as an element of the Lie algebra $\g$}.
\end{definition}

\begin{remark}
The canonical Lie-Poisson structure has the remarkable property to be
\emph{linear}. A Poisson bracket on a vector space is said to be \emph{linear}
if the bracket of two linear functionals is itself a linear functional.
\end{remark}

Each element $\gamma\in\bigwedge^{2}\gstar$ can be viewed
as a Poisson bivector. Indeed, $\brackets{\gamma}{\gamma}=0$ since
$\gamma$ is a constant tensor field. As such, $\gamma$ defines a
Poisson structure on $\gstar$. The condition of
compatibility with the canonical Lie-Poisson structure,
$\brackets{W_{0}}{\gamma}=0$, can be written as (see \cite{Vai94}, Chapter 3)
\begin{equation}\label{equ:compat}
    \gamma(\brackets{u}{v},w)+\gamma(\brackets{v}{w},u) + \gamma(\brackets{w}{u},v)=0,\qquad u,v,w \in \g.
\end{equation}

On a Lie group $G$, a right-invariant $k$-form $\omega$ is completely
defined by its value at the unit element $e$, and hence by an
element of $\bigwedge^{k}\gstar$. In other words, there is
a natural isomorphism between the space of
right-invariant $k$-forms on $G$ and $\bigwedge^{k}\gstar$.
Moreover, since the exterior differential $d$ commutes with right
translations, it induces a linear operator
$\partial:\bigwedge^{k}\gstar \to \bigwedge^{k+1}\gstar$ that satisfies $\partial\circ\partial=0$ and

\begin{enumerate}
\item $\partial\gamma=0$ for $\gamma\in\bigwedge^{0}\gstar=\Real$;
\item $\partial\gamma\,(u,v)=-\gamma(\brackets{u}{v})$ for $\gamma\in\bigwedge^{1}\gstar=\gstar$;
\item $\partial\gamma\,(u,v,w)=\gamma(\brackets{u}{v},w) + \gamma(\brackets{v}{w},u)+\gamma(\brackets{w}{u},v)$ for $\gamma\in\bigwedge^{2}\gstar$,
\end{enumerate}
where $u,v,w \in \g$, as one can check by direct computation
(see \cite{GS90}, Chapter 24). The kernel $Z^n(\g)$ of
$\partial: \bigwedge^n(\gstar) \to \bigwedge^{n+1}(\gstar)$ is the space of \emph{n-cocycles}
and the range $B^n(\g)$ of
$\partial: \bigwedge^{n-1}(\gstar) \to
\bigwedge^{n}(\gstar)$ is the spaces of
\emph{n-coboundaries}. Notice that $B^n(\g)
\subset Z^n(\g)$. The quotient space $H^n_{CE}(\g)=
Z^n(\g)/B^n(\g)$ is the $n$-th \emph{Lie algebra cohomology} or \emph{Chevaley-Eilenberg cohomology
group} of $\g$. Notice that in general the Lie algebra cohomology is different from the de Rham cohomology $H^n_{DR}$. For example,
$H^{1}_{DR}(\Real) = \Real$ but $H^{1}_{CE}(\Real) =0$.

Each  $2$-cocycle $\gamma$ defines a
Poisson structure on $\gstar$ compatible with the canonical one.
Indeed \eqref{equ:compat} can be recast as $\partial \gamma=0$. Notice that
the Hamiltonian vector field $X_{f}$ of a function $f\in
C^{\infty}(\gstar)$ computed with respect to the
Poisson structure defined by the $2$-cocycle $\gamma$ is
\begin{equation}\label{equ:ConstantPoissonHamiltonian}
    X_{f}(m) = \gamma (d_mf, \cdot).
\end{equation}

\begin{definition}
A \emph{modified Lie-Poisson structure} is a Poisson
structure on $\gstar$ whose Poisson bivector is given by
$W_{\gamma} = W_{0} + \gamma,$
where $W_{0}$ is the canonical Poisson bivector and $\gamma$ is a
$2$-cocycle.
\end{definition}

\begin{example}
A special case of modified Lie-Poisson structure is given by a
$2$-cocycle $\gamma$ which is a coboundary. If
$\gamma = \partial m_{0}$
for some $m_{0}\in \gstar$, the expression
\begin{equation*}
    \poisson{f}{g}_{\gamma} (m) = m_{0}(\brackets{d_{m}f}{d_{m}g})
\end{equation*}
looks like if the Lie-Poisson bracket had been "frozen" at a
point $m_{0} \in \gstar$ and for this reason some
authors call it a \emph{"freezing"} structure.
\end{example}


\section{Modified Lie-Poisson structures on $\VectS$}

The group $\DiffS$ of smooth orientation-preserving
diffeomorphisms of the circle $\Circle$ is endowed with a
smooth manifold structure based on the Fr\'{e}chet space
$\CS$. The composition and the inverse are
both smooth maps $\DiffS\times \DiffS \to \DiffS$,
respectively $\DiffS \to \DiffS$, so that $\DiffS$
is a Lie group \cite{Ham82}. Its Lie algebra $\VectS$ is the space
of smooth vector fields on $\Circle$, which is isomorphic to the space
$\CS$ of periodic functions. The Lie bracket
on $\VectS$ is given by
\begin{equation*}
    \brackets{u}{v} = uv_{x} - u_{x}v .
\end{equation*}
Since the topological dual of the Fr\'{e}chet space
$\VectS$ is too big, being isomorphic to the space of
distributions on the circle, we restrict our attention in the
following to the \emph{regular dual} $\VectSstar$, the
subspace of distributions defined by linear
functionals of the form
\begin{equation*}
    u \mapsto \int_{\Circle}mu\, dx
\end{equation*}
for some function $m \in \CS$. The regular
dual $\VectSstar$ is therefore isomorphic to
$\CS$ by means of the $L^2$ inner
product~\footnote{In the sequel, we use the notation $u,v,\dotsc$
for elements of $\VectS$ and $m,n,\dots$ for elements of
$\VectSstar$ to distinguish them, although they all
belong to $\CS$.}
\begin{equation*}
    <u , v > = \int_{\Circle} uv \, dx .
\end{equation*}
Let $f$ be a smooth real valued function on $\CS$.
Its \emph{Fr\'{e}chet} derivative at $m$, $df(m)$ is a
linear functional on $\CS$. We say that $f$
is a \emph{regular function} if there exists a smooth map
$\delta f : \CS \to \CS$
such that
\begin{equation*}
    df(m)\,M = \int_{\Circle} M\cdot\delta f (m)
 \, dx,\qquad m, M \in \CS.
\end{equation*}
That is, the Fr\'{e}chet derivative $df(m)$ belongs to the regular dual
$\VectSstar$ and the mapping $m \mapsto \delta f(m)$ is
smooth. The map $\delta f$ is a vector field on
$\CS$, called the \emph{gradient} of $f$ for the $L^{2}$-metric. In other words, a
regular function is a smooth function on $\CS$ which has a smooth gradient.

\begin{example}
Typical examples of \emph{regular functions} are nonlinear \emph{functionals}
over the space $\CS$, like
\begin{equation*}
    f(m) =  \int_{\Circle} \left( m^{2} + mm_{x}^2 \right)  \, dx \quad\hbox{with}\quad \delta f(m)=2m-m_x^2-2mm_{xx},
\end{equation*}
as well as linear functionals
\begin{equation*}
    f(m) =  \int_{\Circle} um  \,
dx\quad\hbox{with}\quad \delta f(m)=u \in  \CS.
\end{equation*}
Notice that the smooth function $f_\theta: \CS \to \Real$ defined by $f_\theta(m)=m(\theta)$ for some fixed $\theta \in \Circle$ is not regular as $\delta f_\theta$ is the Dirac measure at $\theta$.
\end{example}

Conversely, a smooth vector field $X$ on $\VectSstar$ is
called a \emph{gradient} if there exists a \emph{regular function}
$f$ on $\VectSstar$ such that $X(m) = \delta f(m)$ for all $m
\in \VectSstar$. Observe that if $f$ is a smooth real valued
function on $\CS$ then its second Fr\'{e}chet
derivative is symmetric \cite{Ham82}, that is,
\begin{equation*}
    d^{2}f(m)(M,N) = d^{2}f(m)(N,M), \qquad m,M,N \in
    \CS .
\end{equation*}
For a regular function, this property can be written as
\begin{equation}\label{equ:SymmetricCondition}
    \int_{\Circle} \Big( d\,\delta f(m)M \Big)N  \, dx = \int_{\Circle} \Big( d\,\delta f(m)N \Big)M  \, dx ,
\end{equation}
for all $m,M,N \in \CS$. Hence the linear
operator $d\,\delta f(m)$ is symmetric for the $L^{2}$-inner product on $\CS$ for each $m \in \CS$. We will resume this fact in the following lemma.

\begin{lemma}
A necessary condition for a vector field $X$ on
$\CS$ to be a \emph{gradient} is that its
Fr\'{e}chet derivative $dX(m)$ is a symmetric linear operator.
\end{lemma}

To define a \emph{Poisson bracket} on the space of \emph{regular
functions} on $\VectSstar$, we consider a one-parameter family
of linear operators $J(m)$ and set
\begin{equation}\label{equ:PoissonBracket}
    \poisson{f}{g}(m) = \int_{\Circle} \delta f(m) \, J(m) \, \delta g(m)  \, dx .
\end{equation}
The operators $J(m)$ must satisfy certain conditions in
order for \eqref{equ:PoissonBracket} to be a valid Poisson structure
on $\VectSstar$.

\begin{definition}
A family of linear operators $J(m)$ on $\VectSstar$ defines a Poisson structure on $\VectSstar$ if
~\eqref{equ:PoissonBracket} satisfies
\begin{enumerate}
    \item $\poisson{f}{g}$ is regular if $f$ and $g$ are regular,
    \item $\poisson{g}{f} = - \poisson{f}{g}$,
    \item $\poisson{\poisson{f}{g}}{h} + \poisson{\poisson{g}{h}}{f}
    + \poisson{\poisson{h}{f}}{g} = 0$.
\end{enumerate}
\end{definition}

Notice that the second condition above simply
means that $J(m)$ is a skew-symmetric operator for each $m$.

\begin{example}
The canonical Lie-Poisson structure on $\VectSstar$ given by
\begin{equation*}
    \poisson{f}{g}(m) = m \left( \brackets{\delta f}{\delta g} \right)
    = \int_{\Circle} \delta f (m) \left( mD + Dm \right) \delta
    g(m)\,dx
\end{equation*}
is represented by the one-parameter family of
skew-symmetric operators
\begin{equation}\label{equ:LiePoissonOperator}
    J(m)=mD+Dm
\end{equation}
where $D=\partial_x$. It can be checked that all the
three required properties are
satisfied. In particular, we have
\begin{equation*}
    \delta \poisson{f}{g} = d \, \delta f (J \delta g) -
d \, \delta g (J \delta
    f) + \delta f \, \delta g_{x} - \delta g \, \delta f_{x} .
\end{equation*}
\end{example}

\begin{definition}
The \emph{Hamiltonian} of a \emph{regular} function $f$, for a
Poisson structure defined by $J$ is defined as the vector field
\begin{equation*}
    X_{f}(m) = J(m) \, \delta f(m).
\end{equation*}
\end{definition}

\begin{proposition}\label{prop:LinearHamiltonianCriterium}
A necessary condition for a smooth vector field $X$ on
$\VectSstar$ to be \emph{Hamiltonian} with respect to the
Poisson structure defined by a \emph{constant} linear operator $K$
is the symmetry of the operator $dX(m)\circ K$ for each $m\in \VectSstar$.
\end{proposition}

\begin{proof}
If $X$ is Hamiltonian, we can find a regular function $f$ such that
\begin{equation*}
    X(m) = K \delta f(m) .
\end{equation*}
Moreover, since $K$ is a constant linear operator, we have
\begin{equation*}
    d\big( K\, \delta f \big)(m) \, M = K \circ \big( d\delta f(m) \big) \, M .
\end{equation*}
Therefore, we get
\begin{align*}
    < { dX(m)\circ K \, M , \, N }  >
        & = < { K \circ d\delta f(m) \circ K \, M , \, N }  > \\
        & = < { M , \, K \circ d\delta f(m) \circ K \, N }  > \\
        & = < { M , \, dX(m)\circ K \, N }  > ,
\end{align*}
since $K$ is skew-symmetric and $d\delta f(m)$ is symmetric.
\end{proof}

A $2$-cocycle on $\VectS$ is a bilinear functional
$\gamma$ represented by a skew-symmetric operator $K:
\CS \to \CS$ such that
\begin{equation*}
    \gamma(u,v) = < u,Kv  > = \int_{\Circle} u\, K
    \, v \, dx ,
\end{equation*}
and satisfying the Jacobi identity
\begin{equation*}
    < \brackets{u}{v}, Kw  > + < \brackets{v}{w}, Ku
 > +
    < \brackets{w}{u}, Kv  > = 0 .
\end{equation*}
If $K$ is a differential operator we call $\gamma$ a
\emph{differential cocycle}. Gelfand and Fuks \cite{GF68} observed
that all differential $2$-cocycles of $\VectS$ belong
to the one-dimensional cohomology class generated by $[D^{3}]$.
Moreover, each regular $2$-coboundary is represented by the
skew-symmetric operator
\begin{equation*}
    m_{0}D + Dm_{0} ,
\end{equation*}
for some $m_{0} \in \CS$. Therefore, each
differential $2$-cocycle of $\VectS$ is represented by
an operator of the form
\begin{equation}\label{equ:PoissonOperator}
    K = m_{0}D + Dm_{0} + \beta D^{3}
\end{equation}
where $m_{0} \in \CS$ and $\beta \in
\Real$ (see also \cite{GR07}).

For $k\geq0$ and $u,v\in \VectS\equiv
\CS$, let us now define the $H^k$ (Sobolev)
inner product by
\begin{equation*}
    < u,\, v >_{k} = \int_{\Circle} \sum_{i=0}^{k} (\partial_{x}^{i}u)\,(\partial_{x}^{i}v)\, dx = \int_{\Circle}A_{k}(u)\, v\, dx\,,
\end{equation*}
where
\begin{equation}\label{equ:Ak}
    A_{k} = 1-\frac{d^{2}}{dx^{2}}+...+(-1)^{k}\frac{d^{2k}}{dx^{2k}}
\end{equation}
is a continuous linear isomorphism of $\CS$.
Note that $A_{k}$ is a symmetric operator for the $L^{2}$ inner
product since
\begin{equation*}
    \int_{\Circle} A_{k}(u) \, v\, dx = \int_{\Circle}u\, A_{k}(v)\, dx.
\end{equation*}
The operator $A_{k}$ gives rise to a Hamiltonian function on
$\VectSstar$ given by
\begin{equation*}
    h_{k}(m) = \int_{\Circle} \tfrac{1}{2} \, m(A_{k}^{-1}m) \, dx .
\end{equation*}
The corresponding Hamiltonian vector field $X_{k}$ is given by
\begin{equation*}
    X_{k}(m) = (mD + Dm)(A_{k}^{-1}m) = 2mu_{x} + um_{x} ,
\end{equation*}
if we let $m = A_{k}u$.

\begin{theorem}
The Hamiltonian vector field $X_{k}$ is bi-Hamiltonian relatively to
a modified Lie-Poisson structure if and only if $k \in \{0,1\}$.
\end{theorem}

\begin{proof}
It is well known (see \cite{Mor98}) that $X_0$ is bi-Hamiltonian
with respect to the operator $D$ which represents a coboundary.
It is also known that $X_{1}$ is a bi-Hamiltonian vector field with
respect to the cocycle represented by the operator
$D(1-D^{2})$ cf. \cite{CH93,CM99,FF81}. Notice that this cocycle is not a
coboundary.

We will now show that there is no differential cocycle
\begin{equation*}
    K = m_{0}D + Dm_{0} + \beta D^{3}
\end{equation*}
for which $X_{k}$ could be Hamiltonian unless $k \in \{0,1\}$. We have
\begin{equation*}
    dX_{k} (m) = 2u_{x}I + uD + 2m DA_{k}^{-1} + m_{x}A_{k}^{-1} ,
\end{equation*}
and in particular, for $m=1$,
\begin{equation*}
    dX_{k} (1) = D + 2DA_{k}^{-1} .
\end{equation*}
Letting
\begin{equation*}
    P(m) = dX_{k}(m) \circ K ,
\end{equation*}
we obtain that
\begin{equation*}
    P(1) = \big( D + 2DA_{k}^{-1} \big) \circ \big( m_{0}D + Dm_{0} \big) + \beta D^{4} ( 1 + 2A_{k}^{-1}) ,
\end{equation*}
whereas
\begin{equation*}
    P(1)^{\ast} = \big( m_{0}D + Dm_{0} \big) \circ \big( D + 2DA_{k}^{-1} \big) + \beta D^{4} ( 1 + 2A_{k}^{-1}) .
\end{equation*}
Therefore, denoting $m_0'=\partial_x m_0$, we have
\begin{multline*}
    P(1) - P(1)^{\ast} = \big( m_0' D + Dm_0' \big)
     + 2 \big( A_{k}^{-1}Dm_{0}D - Dm_{0}DA_{k}^{-1} \big) + \\
     + 2 \big( A_{k}^{-1}D^{2}m_{0} - m_{0}D^{2}A_{k}^{-1} \big) .
\end{multline*}
If this operator is zero, we must have in particular the relation
\begin{equation*}
    A_{k} \big( P(1) - P(1)^{\ast} \big) A_{k}(e^{irx}) = 0 ,
\end{equation*}
for all $r\in \mathbb{Z}$. But, for $r \neq \pm 1$,
\begin{equation*}
    A_{k}(e^{irx}) = f_{k}(r)\,e^{irx} \quad \text{with} \quad f_{k}(r) = \frac{r^{2k+2}-1}{r^{2}-1} ,
\end{equation*}
and
\begin{equation*}
    A_{k} \big( P(1) - P(1)^{\ast} \big) A_{k}(e^{irx})
\end{equation*}
is of the form $e^{irx}$ times
a polynomial expression in $r$ with highest order term
$2i\,m_0'(x) \, r^{4k+1}$. Therefore, a necessary condition for $X_{k}$
to be Hamiltonian relatively to the Poisson operator $K$ defined
by~\eqref{equ:PoissonOperator} is that $m_{0}$ is a constant.

Let $\alpha = 2m_{0}\in \Real$. Then
\begin{multline*}
    P(m) = dX_{k}(m) \circ K = \alpha \left\{ 2u_{x}D + uD^{2} + 2m D^{2}A_{k}^{-1} + m_{x}DA_{k}^{-1}
    \right\} + \\
    + \beta \left\{ 2u_{x}D^{3} + uD^{4} + 2m D^{4}A_{k}^{-1} + m_{x}D^{3}A_{k}^{-1} \right\}
\end{multline*}
because $D$ and $A_{k}$ commute. By virtue of
Proposition~\ref{prop:LinearHamiltonianCriterium}, a necessary
condition for $X_{k}$ to be Hamiltonian with respect to the cocycle
represented by $K$ is that $P(m)$ is symmetric. That is
\begin{equation}\label{equ:SymmetricOperator}
    < P(m)M,N  > = < M, P(m)N  > ,
\end{equation}
for all $m,M,N \in \CS$. Since this last
expression is tri-linear in the variables $m,M,N$, the equality can
be checked for complex periodic functions $m,M,N$ where the $L^{2}$
inner product is extended naturally into a complex bilinear
functional. That is, the extension is not a hermitian product, we
just allow homogeneity with respect to the complex scalar field in
both components. Let $m = A_{k}u$, $u=\exp(iax)$, $M = \exp(ibx)$
and $N = \exp(icx)$ with $a,b,c \in \mathbb{Z}$. We have
\begin{multline*}
    < P(m)M, N  > = \Big[ (2ab^{3} + b^{4})\beta - (2ab + b^{2})\alpha + \\
    + \Big( (ab^{3} + 2b^{4})\beta - (ab + 2b^{2})\alpha
    \Big)\frac{f_{k}(a)}{f_{k}(b)} \Big] \int_{\Circle} e^{i(a+b+c)x}dx \, ,
\end{multline*}
whereas
\begin{multline*}
    < M, P(m)N  > = \Big[ (2ac^{3} + c^{4})\beta - (2ac +
c^{2})\alpha + \\
    + \Big( (ac^{3} + 2c^{4})\beta - (ac + 2c^{2})\alpha
    \Big)\frac{f_{k}(a)}{f_{k}(c)} \Big] \int_{\Circle} e^{i(a+b+c)x}dx \, .
\end{multline*}
For $a=n$, $b=-2n$ and $c=n$, we obtain
\begin{equation}\label{equ:explicit_expression}
    < P(m)M, N  > = (24n^{4}\beta -6n^{2}\alpha)\frac{f_{k}(n)}{f_{k}(2n)},\
    < M, P(m)N  > = 6n^{4}\beta - 6n^{2}\alpha.
\end{equation}
The equality of the two expressions in \eqref{equ:explicit_expression} for all $n \in \mathbb{N}$ is ensured
by means of \eqref{equ:SymmetricOperator}. For $k=1$ this leads to the
condition $\alpha+\beta=0$ and we recover the second Poisson structure
given by $K=D-D^3$ for which $X_1$ is known to be Hamiltonian with
Hamiltonian function
\begin{equation*}
    \tilde{h}_1(m) = \frac{1}{2}\int_{\Circle}\Big( (A_1^{-1}m)^3+(A_1^{-1}m)\,[(A_1^{-1}m)_x]^2 \Big)\,dx.
\end{equation*}

In the general case, if $\beta \neq 0$, the leading term with respect
to $n$ in the first expression in \eqref{equ:explicit_expression} is $(-48\,\beta\,2^{-2k})$,
whereas in the second it is $(-12\,\beta)$. Thus unless $\beta=0$ we
must have $k=1$. On the other hand, if $\beta=0$, from \eqref{equ:SymmetricOperator}-\eqref{equ:explicit_expression}
we infer that $\alpha f_k(n)=\alpha f_k(2n)$ for all $n \in \mathbb{N}$. Thus $\alpha=0$ unless $k=0$. For $k=0$ we recover the Poisson structure given by $K=D$ for which $X_0$ is Hamiltonian with
Hamiltonian function
\begin{equation*}
    \tilde{h}_0(m)={1 \over 2}\int_{\Circle} m^3\,dx.
\end{equation*}
This completes the proof.
\end{proof}


\section{Conclusion}

We showed that among all $H^k$ Sobolev inner products on
$\CS$, only for $k \in \{0,1\}$ is the associated vector field bi-Hamiltonian relatively to a modified
Lie-Poisson structure. Endowing $\DiffS$ with the $H^1$ right-invariant metric, the associated geodesic equation turns out to be the Camassa-Holm equation \cite{Kou99} (see also \cite{Kol04})
\begin{equation*}
    u_t+uu_x+\partial_x(1-\partial_x^2)^{-1}(u^2+{1 \over
  2}\,u_x^2)=0,
\end{equation*}
a model for shallow water waves (see \cite{CH93} and the alternative derivations in \cite{Con01,Fok95,FL96,Joh02}) that
accommodates waves that exist indefinitely in time \cite{Con97,CE98a} as
well as breaking waves \cite{CE98,CE00}. The bi-Hamiltonian structure
is reflected in the existence of infinitely many conserved integrals
for the equation \cite{CH93,CM99,FF81,Len04a} which are very useful in
the qualitative analysis of its solutions. Both global existence
results and blow-up results can be obtained using certain conservation
laws \cite{Con97,CE98a,Wah05}, while the proof of stability of traveling wave solutions
relies on the specific form of some conserved quantities \cite{Con98,CM99,CS00,Len04}. On the other hand, the geodesic equation on $\DiffS$ for the $L^2$ right-invariant metric is the inviscid Burgers equation
\begin{equation*}
    u_t+3uu_x=0.
\end{equation*}
This model of gas dynamics has been thoroughly studied
(see \cite{CK02} and references therein). In contrast to the case of
the $H^1$ right-invariant metric \cite{CK03}, the Riemannian
exponential map is not a $C^1$ local diffeomorphism in the case of the
$L^2$ right-invariant metric \cite{CK02}. This means that of the two
bi-Hamiltonian vector fields $X_0$ and $X_1$, the second generates a
flow on $\DiffS$ with properties that parallel those of geodesic
flows on finite-dimensional Lie groups.


\bibliographystyle{plain}
\bibliography{CK2006}

\end{document}